	\documentclass[sn-mathphys]{sn-jnl}
	
	\jyear{2021}%
	
	\theoremstyle{thmstyleone}%
	\newtheorem{theorem}{Theorem}
	\newtheorem{proposition}[theorem]{Proposition}%

	\theoremstyle{thmstyletwo}%
	\newtheorem{example}{Example}%
	\newtheorem{remark}{Remark}%
	\newtheorem{lemma}{Lemma}%
	\theoremstyle{thmstylethree}%
	\newtheorem{definition}{Definition}%
	
\begin{document}
	
	\title[Article Title]{On construction of quantum codes with dual-containing quasi-cyclic codes}
	
	
	\author*[1]{\fnm{Chaofeng} \sur{Guan}}\email{gcf2020yeah.net}
	
	\author[1]{\fnm{Ruihu} \sur{Li}}\email{liruihu@aliyun.com}
	\equalcont{These authors contributed equally to this work.}
	
	\author[1]{\fnm{Liangdong } \sur{Lu}}\email{kelinglv@163.com}
	\equalcont{These authors contributed equally to this work.}
	
	\author[1]{\fnm{Yang } \sur{Liu}}\email{liu\_yang10@163.com}
	\equalcont{These authors contributed equally to this work.}
	
	\author[1]{\fnm{Hao } \sur{Song}}\email{songhao\_kgd@163.com}
	\equalcont{These authors contributed equally to this work.}
	
	\affil*[1]{\orgdiv{Fundamentals Department}, \orgname{Air Force Engineering University}, \orgaddress{ \city{Xi'an}, \state{Shaanxi}, \postcode{710051}, \country{ P. R. China}}}

	
	\abstract{
	One of the main objectives of quantum error-correction theory is to construct quantum codes with optimal parameters and properties. 
	In this paper, we propose a class of 2-generator quasi-cyclic codes and study their applications in the construction of quantum codes over small fields. 
	Firstly, some sufficient conditions for these 2-generator quasi-cyclic codes to be dual-containing concerning Hermitian inner product are determined. 
	Then, we utilize these Hermitian dual-containing quasi-cyclic codes to produce quantum codes via the famous Hermitian construction.
	Moreover, we present a lower bound on the minimum distance of these quasi-cyclic codes, which is helpful to construct quantum codes with larger lengths and dimensions.
	As the computational results, many new quantum codes that exceed the quantum Gilbert-Varshamov bound are constructed over $F_q$, where $q$ is $2,3,4,5$. In particular, 16 binary quantum codes raise the lower bound on the minimum distance in Grassl's table \cite{Grassl:codetables}. In nonbinary cases, many quantum codes are new or have better parameters than those in the literature.
	}
	
	
	
	
	\maketitle
	\section{Introduction}\label{sec1}
	
	Quantum error-correction codes (QECCs) play a prominent role in the implementation of quantum computing by virtue of they can protect fragile qubits from noises (decoherence).
	Since the initial work of Shor \cite{shor1995scheme} and Steane \cite{Steane1996}, the theory of QECCs has developed rapidly.
	With the efforts of scholars, the connection between QECCs and classical codes has been gradually established, so that the problem of constructing QECCs is transformed into the problem of constructing self-orthogonal or dual-containing classical codes over $F_q$ or $F_{q^2}$ under different inner products \cite{Calderbank1996,Calderbank1998,Ashikhmin2001,Ketkar2006}.

	In \cite{Grassl:codetables}, Grassl et al. collected the best-known binary QECCs and founded an online code table, which is considered to be the most widely used and most challenging to break one.
	In addition, Edel et al. \cite{Edel2022} cataloged a discrete but more involved online code table, which is an essential comparison for nonbinary QECCs.
	In this paper, we will provide many QECCs that have better parameters than those in \cite{Grassl:codetables} and \cite{Edel2022}.
	
	Quasi-cyclic codes are a natural extension of cyclic codes, with rich algebraic structure and excellent properties.  In \cite{kasami1974gilbert}, Kasami et al. proved that quasi-cyclic codes satisfy the modified Gilbert–Varshamov bound, i.e., quasi-cyclic codes are asymptotically good.
	In addition, a mass of record-breaking classical codes were constructed from quasi-cyclic codes
	\cite{siap2000new,daskalov2003new,chen2018some,akre2021generalization}. 
	In \cite{hagiwara2011spatially},  Hagiwara et al. first manipulated quasi-cyclic codes to construct quantum LDPC codes with long lengths.
	In 2018, Galindo et al. \cite{galindo2018quasi} proposed an original method to construct quantum codes employing quasi-cyclic codes, which attracted many authors to devote themselves to the investigation of quasi-cyclic codes to construct QECCs. Thereby, a series of record-breaking QECCs were constructed by quasi-cyclic codes \cite{ezerman2019good,lv2019new,Lv2019newnon,lv2020quantum,lv2020constructions,lv2019quantum,lv2020explicit,guan2021new,yao2021new}.
	
	Inspired by the above work, we propose a new method for constructing QECCs via quasi-cyclic codes.
	This paper is organized as follows.
	In Section 2, fundamentals of linear codes and  QECCs are introduced.
	Section 3 presents a class of 2-generator quasi-cyclic codes and their sufficient conditions for dual-containing under Hermitian inner product. In addition, we derive a lower bound on the minimum Hamming distance for the relevant codes. As applications, many new quantum codes over small finite fields are constructed. 
	Conclusions are given in Section 4.
	\section{Preliminaries}\label{sec2}
	
	A linear code $\mathcal{C}$ of length $n$ over $F_{q^2}$ is an non-empty subset of  $F_{q^2}^{n}$, and can be denoted as $[n,k,d]_{q^2}$.
	Let $\vec{u}=(u_{0},  \ldots, u_{n-1})$ and 
	$\vec{v}=(v_{0},\ldots, v_{n-1})$ be vectors in $ F_{q^2}^{n}$, Hermitian inner product of them can be defined as 
	$\langle\vec{u}, \vec{v}\rangle_{h}=\sum_{i=0}^{n-1}\left(u_{i} v^q_{i}\right)$. 
	The weight of $\vec{u}$, denoted by $wt(\vec{u})$, is the number of nonzero coordinates in $\vec{u}$.
	The minimum Hamming distance of $\mathcal{C}$ is $d(\mathcal{C})=\min \left\{wt(\vec{u}) \mid \vec{u} \in \mathcal{C}\right\}.$
	The Hermitian dual code of $\mathcal{C}$ can be denoted as 
	$\mathcal{C}^{\perp_{q}}=\{\vec{v} \in F_{q^2}^{n} \mid\langle\vec{u}, \vec{v}\rangle_{h}=0, \forall \vec{u} \in \mathcal{C}\}$, and parameters of $\mathcal{C}^{\perp_{q}}$ is $[n,n-k,d^\perp]_{q^2}$.
	If $\mathcal{C} \subset C^{\perp_{q}}$, then we can say $\mathcal{C}$ is a Hermitian self-orthogonal code and $\mathcal{C}^{\perp_{q}}$ is a Hermitian dual-containing code.

	If $\mathcal{C}$ is closed under a cyclic shift, i.e.,  for any $c=\left(c_{0}, c_{1}, \ldots, c_{n-1}\right)\in \mathcal{C}$, there will 
	$c^\prime =\left(c_{n-1}, c_{0}, \ldots, c_{n-2}\right)\in \mathcal{C}$,  then we can say $\mathcal{C}$ is a cyclic code. 
	Define the quotient ring  $\mathcal{R}=F_{2}[x] /\left\langle x^{n}-1\right\rangle$.
	If $\mathcal{C}$ is generated by a monic divisor  $g(x)$  of  $x^{n}-1$, i.e., $\mathcal{C}=\langle g(x)\rangle$ and $g(x)\mid x^n-1$, then $g(x)$ is called generator polynomial of $\mathcal{C}$.
	Let $gcd(n, q) = 1$ and $ \Omega_{n}=\{0,1, \ldots, n-1\}$  and  $\varsigma $  be a primitive  $n$-th root of unity in some extended fields of  $F_{q^2}$.  
	Then defining set of $\mathcal{C}=\langle g(x)\rangle$  can be written as $T=\left\{i \in  \Omega_{n} \mid g\left(\varsigma  ^{i}\right)=0\right\}$. 
	If $C_{i}$ is a $q^2$-cyclotomic coset modulo  $n$, then it can be denoted as
	$C_{i}=\{i, i q^2, \ldots, i q^{2(s-1)} \mid i \in \Omega_{n}\}$, where  $s$  is the smallest positive integer satisfied with  $i q^{2s} \equiv i \bmod n$.
	For each $i \in \Omega_{n}$, the cyclotomic cosets $C_i$ is skew symmetric if $n-qi\in C_i$; otherwise it is skew asymmetric. Skew asymmetric cosets $C_i$ and $C_{n-qi}$ occur in pairs and are called skew asymmetric pairs, abbreviated as $(C_i, C_{n-qi})$. If  $T \cap T^{-q}=\emptyset$ and any two cosets in  $T$  cannot form a skew asymmetric pair, then  $\mathcal{C}^{\perp_{h}} \subseteq \mathcal{C}$, i.e., $g(x)\mid g^{\perp_q}(x)$.

	If a cyclic shift of any codeword of $\mathcal{C}$ by $n$ positions is also codeword, then $\mathcal{C}$ is a quasi-cyclic code. The length of $\mathcal{C}$ is $nl$.  
	The generator matrix of quasi-cyclic code is composed of circulant matrices. An $n\times n$ circulant matrix $M$ is defined as
	
	$$
	M=\left(\begin{array}{ccccc}
	m_{0} & m_{1} & m_{2} & \ldots & m_{n-1} \\
	m_{p-1} & m_{0} & m_{1} & \ldots & m_{n-2} \\
	\vdots & \vdots & \vdots & \vdots & \vdots \\
	m_{1} & m_{2} & m_{3} & \ldots & m_{0}
	\end{array}\right).$$

	The generator matrix of a 2-generator quasi-cyclic code with index $2$ can be transformed into rows of $n \times n$ circulant matrices by suitable permutation of columns, which has the following form
	$$
	G=\left(\begin{array}{cccc}
	M_{1,1} & M_{1,2}\\
	M_{2,1} & M_{2,2}\\
	\end{array}\right),$$
	where  $M_{i, j}$  is circulant matrices determined by polynomial  $m_{i, j}(x)$, where  $1 \leq i \leq 2$  and  $1 \leq j \leq 2$. 
	
	
	A QECC $\mathcal{Q}$ of length $n$ is a $K$-dimensional subspace of $q^n$-dimensional Hilbert space $(\mathbb{C}^{q})^{\otimes n}$, where $\mathbb{C}$ represents complex field and $(\mathbb{C}^{q})^{\otimes n}$ is the $n$-fold tensor power of $\mathbb{C}^q$. 
	Three basic parameters can be used to describe $\mathcal{Q}$: length $n$, dimension $k$ and minimum distance $d$, so that $\mathcal{Q}$ can be denoted as $[[n,k,d]]_{q}$, where $k=log_{q}K$.
	With \cite{Ketkar2006}, there is a relationship between QECCs and classical dual-containing linear codes under Hermitian inner product. 
	
	\begin{lemma}\label{construction}
	(\cite{Ketkar2006}, Hermitian construction) A Hermitian dual-containing  $[n, k]_{q^{2}}$  linear code  $\mathcal{C}$  such that there are no vectors of weight less than  $d$  in  $\mathcal{C} \backslash C^{\perp_{h}}$  yields a pure QECC with parameters  $[[n, 2 k-n, d]]_{q}$. 
	\end{lemma}
	
	Typically, QECCs can also be derived from existing ones by the following propagation rules, which will be utilized later.
	
	\begin{lemma}\label{DerivativeCodes}\label{propagation_rule}
	(\cite{Ketkar2006}, propagation rules) If an  $[[n, k, d]]_{q}$ pure QECC exists. Then the following QECCs also exist.\\
	(1)  $[[n, k-1, d]]_{q}$ for $k\ge 1$;\\
	(2)  $[[n+1, k, d]]_{q}$ for $k>0$.
	\end{lemma}
	
	Like the classical counterpart, one of the central tasks of quantum error correction is to construct QECCs with suitable parameters. When length $n$ and dimension $k$ are fixed, we want to obtain a sizable minimum distance $d$. Conversely, when minimum distance $d$ is fixed, we want the rate $\frac{k}{n}$ to be larger.
	Some bounds are helpful to judge the performance of QECCs, among which the quantum Gilbert-Varshamov (GV) bound \cite{feng2004finite} is widely used. 
	\begin{lemma}
	(\cite{feng2004finite}, quantum GV bound) Let $n> k_{GV}\ge 2$, and $n\equiv k_{GV} \left ( \bmod \ 2 \right )$, $d\ge 2$.
	Then, there exists an $[[n, k_{GV}, d]]_q$ pure QECC if the inequality is met.
	
	\begin{equation}\label{eqGV}\frac{q^{n-k_{GV}+2}-1}{q^{2}-1}>\sum_{i=1}^{d-1}\left(q^{2}-1\right)^{i-1}\left(\begin{array}{c}n \\i\end{array}\right).\end{equation}
	\end{lemma}
	
	It is easy to determine that the quantum GV bound illustrates the existence of a class of quantum codes, and quantum codes satisfying inequality (\ref{eqGV}) exist. By virtue of \cite{feng2004finite}, if $k\ge k_{GV}$ holds for certain n and d, i.e. QECC $[[n,k,d]]_q$  beats this bound, then it is considered to have excellent parameters.

	\section{A method for producing QECCs}\label{sec3}
	In this section, a family of 2-generator quasi-cyclic codes is presented. We derive the algebraic form of its Hermitian dual codes, which yields sufficient conditions for those quasi-cyclic codes to be dual-containing under Hermitian inner product. In addition, we also deduce a lower bound on the minimum Hamming distance of them. As an application, many new QECCs are constructed via the Hermitian construction.

	Let $g(x)=g_{0}+g_{1} x+g_{2} x+\cdots+   g_{n-1} x^{n-1} \in \mathcal{R}$ and $[g(x)]=[g_{0},g_{1},g_{2},\cdots,g_{n-1}]$  represents vectors in  $F_{q^2}^{n}$  determined by the coefficient of $g(x)$ in an ascending order. We also define the following polynomials:
	$$\bar{g}(x)=g_{0}+g_{n-1} x+g_{n-2} x^{2}+ \cdots+g_{1} x^{n-1},$$
	$$g^{q}(x)=g^q_{0}+g^q_{1} x+g^q_{2} x+\cdots+   g^q_{n-1} x^{n-1}.$$
	
	Moreover, let  $ h(x)=(x^{n}-1)/g(x)$, then $g^{\perp}(x)=x^{\operatorname{deg}(h(x))} h\left(\frac{1}{x}\right)$. The Hermitian dual code of $\left \langle g(x) \right \rangle $ is $\left \langle g^{\perp_q}(x) \right \rangle $. 
	
	\begin{definition}\label{def1} Let  $\mathcal{C}_{q^2}(g_1,g_2,t)$  be a quasi-cyclic code over  ${F}_{q^2}$  of length  $2n$  generated by  $\left([t(x)g_{1}(x)], [g_{1}(x)]\right)$  and  $\left([g_{2}(x)], [t(x) g_{2}(x)]\right)$, where $g_{1}(x)$, $g_{2}(x)$  and $t(x)$ are polynomials in $\mathcal{R}$ such that $g_{1}(x) \mid \left(x^{n}-1\right)$, $g_{2}(x) \mid\left(x^{n}-1\right)$. 
	Then a generator matrix $G$ of $\mathcal{C}_{q^2}(g_1,g_2,t)$ will have the following form:
	$$G=\left(\begin{array}{cc}G_{T1}& G_1 \\G_2 & G_{T2}\end{array}\right),$$ where  $G_{1}$  and  $G_{2}$  are generator matrices of cyclic codes  $\langle g_1(x)\rangle$  and  $\langle g_2(x)\rangle$, respectively. $G_{T1}$ and $G_{T2}$ are $(n-deg(g_1(x))) \times n$ and $(n-deg(g_2(x))) \times n$ circulant matrices determined by  $[t(x)g_1(x)]$ and $[t(x)g_2(x)]$, separately.
	\end{definition}
	
	\begin{definition} Let  $\mathcal{C}_{0}$  be a quasi-cyclic code over  ${F}_{q^2}$  of length  $2n$, and its generator polynomial pairs are  $( [-\bar{t}^q(x) g_{1}^{\perp_q}(x)], [g_{1}^{\perp_q}(x)])$  and  $( [g_{2}^{\perp_q}(x)],[- \bar{t}^q(x) g_{2}^{\perp_q}(x)])$. 
	Then a generator matrix $G_0$ of $\mathcal{C}_{0}$ will have the following form:
	$$G_0=\left(\begin{array}{cc}
	G^{\perp}_{\bar{T1}}& G^{\perp}_1 \\
	G^{\perp}_2 & G^{\perp}_{\bar{T2}}
	\end{array}\right),$$
	where  $G^{\perp}_{1}$ and $G^{\perp}_{2}$  are generator matrices of cyclic codes  $\langle g^{\perp_q}_1(x)\rangle$  and $\langle g^{\perp_q}_2(x)\rangle$, respectively. $G^{\perp}_{\bar{T1}}$ and $G^{\perp}_{\bar{T2}}$ are $deg(g_1(x)) \times n$ and $deg(g_2(x)) \times n$ circulant matrices determined by  $[-\bar{t}^q(x) g_{1}^{\perp_q}(x)]$ and $[-\bar{t}^q(x) g_{2}^{\perp_q}(x)]$, separately.
	\end{definition}

	\begin{lemma}(\cite{Bierbrauer2014TheSO}, Lemma 1)\label{Hermitianself-orthogonal} Let $\mathcal{C}$ be a linear code over $F_{q^2}^n$, then $\mathcal{C}$ is self-orthogonal concerning the Hermitian inner product if and only if $\langle c,c \rangle_h=0$ for all codewords $c$ of $\mathcal{C}$.\end{lemma}
	
	\begin{lemma}(\cite{lv2020quantum}, Proposition 1)\label{exchange_law}
	Let  $f(x)$, $g(x)$  and  $h(x)$  be monic polynomials in  $\mathcal{R}$. Then the following equality of Hermitian inner product of vectors in  $F_{q^{2}}^{n}$  holds:
	$$\langle[f(x) g(x)],[h(x)]\rangle_{h}=\left\langle[g(x)],\left[\overline{f^{q}}(x) h(x)\right]\right\rangle_{h}.$$
	\end{lemma}
	
	\begin{theorem}\label{inner_product_zero}
	Let  $\mathcal{C}_{f}$ and $\mathcal{C}_{g}$ are linear codes with length $n$, generated by $f(x)$ and $g(x)$, respectively. The sufficient and necessary conditions for $\langle[a(x)f(x)],[b(x)g(x)]\rangle_{h}=0$ to hold for any polynomials $a(x)$ and $b(x)$ in $\mathcal{R}$ are $g^{\perp_q} (x)\mid f(x)$ and $f^{\perp_q} (x)\mid g(x)$.
	\end{theorem}
	\begin{proof}
	It is easy to deduce that if $g^{\perp_q} (x)\mid f(x)$ and $f^{\perp_q} (x)\mid g(x)$, then $\mathcal{C}_{f}\subset\mathcal{C}^{\perp_q}_{g}$ and $\mathcal{C}_{g}\subset\mathcal{C}^{\perp_q}_{f}$, so $\langle[a(x)f(x)],[b(x)g(x)]\rangle_{h}=0$. 	
	
	Similarly, according to the definition of Hermitian dual code, if $\langle[a(x)f(x)],[b(x)g(x)]\rangle_{h}=0$ holds for any polynomials $a(x)$, $b(x)\in \mathcal{R}$, then $\mathcal{C}_{f}\subset\mathcal{C}^{\perp_q}_{g}$ and $\mathcal{C}_{g}\subset\mathcal{C}^{\perp_q}_{f}$, so there is $g^{\perp_q} (x)\mid f(x)$ and $f^{\perp_q} (x)\mid g(x)$.
	Therefore, we have finished the proof.
	\end{proof}
	
	In addition, the condition in Theorem \ref{inner_product_zero} can be simplified as  $g^{\perp_q} (x)\mid f(x)$ for reason that $g^{\perp_q} (x)\mid f(x)$ and $f^{\perp_q} (x)\mid g(x)$ are equivalent.
	
	\begin{proposition}\label{dual_code} If t (x) can be such that parameters of $\mathcal{C}_{q^2}(g_1,g_2,t)$ and $\mathcal{C}_0$ are $\left[2 n, 2 n-\operatorname{deg}\left(g_{1}(x)\right)-\operatorname{deg}\left(g_{2}(x)\right)\right]$ and $\left[2 n, \operatorname{deg}\left(g_{1}(x)\right)+\operatorname{deg}\left(g_{2}(x)\right)\right]$, respectively. Then the Hermitian dual code of $\mathcal{C}_{q^2}(g_1,g_2,t)$ is $\mathcal{C}_0$.
	\end{proposition}
	
	\begin{proof}
	Set $c_1=([a(x) t(x) g_{1}(x)+b(x) g_{2}(x)],[a(x)g_{1}(x)+b(x) t(x) g_{2}(x)])$,
	$c_2=([-c(x)\bar{t}^q(x) g_{1}^{\perp_q}(x)+d(x)g_{2}^{\perp_q}(x)],[c(x)g_{1}^{\perp_q}(x)-d(x)\bar{t}^q(x)g_{2}^{\perp_q}(x)])$, where $a(x)$, $b(x)$, $c(x)$, $d(x)$ are arbitrary polynomials in $\mathcal{R}$.
	Then, any codewords in $\mathcal{C}_{q^2}(g_1,g_2,t)$ and $\mathcal{C}_0$ can be represented by $c_1$, $c_2$, respectively.
	The Hermitian inner product of them $\langle c_1, c_2\rangle_{h}$ is equal to
	$$ \begin{aligned}
	\langle[a(x)t(x)g_{1}(x)+b(x) g_{2}(x)],
	[-c(x)\bar{t}^q(x) g_{1}^{\perp_q}(x)+d(x)g_{2}^{\perp_q}(x)]\rangle_{h} \\
	+\langle[a(x)g_{1}(x)+b(x) t(x) g_{2}(x)],[c(x)g_{1}^{\perp_q}(x)-d(x)\bar{t}^q(x)g_{2}^{\perp_q}(x)]\rangle_{h} \\
	=\langle[a(x)t(x)g_{1}(x)],[-c(x)\bar{t}^q(x) g_{1}^{\perp_q}(x)]\rangle_{h} 
	+\langle[a(x)t(x)g_{1}(x)],[d(x)g_{2}^{\perp_q}(x)]\rangle_{h}  \\
	+\langle[b(x) g_{2}(x)],[-c(x) \bar{t}^q(x)g_{1}^{\perp_q}(x)]\rangle_{h} 
	+\langle[b(x) g_{2}(x)],[d(x)g_{2}^{\perp_q}(x)]\rangle_{h}\\
	+\langle[a(x)g_{1}(x)],[c(x)g_{1}^{\perp_q}(x)]\rangle_{h} 
	+\langle[a(x)g_{1}(x)],[-d(x)\bar{t}^q(x)g_{2}^{\perp_q}(x)]\rangle_{h} \\
	+\langle[b(x) t(x) g_{2}(x)],[c(x)g_{1}^{\perp_q}(x)]\rangle_{h} 
	+\langle[b(x) t(x) g_{2}(x)],[-d(x)\bar{t}^q(x)g_{2}^{\perp_q}(x)]\rangle_{h}. \\
	\end{aligned}$$
	
	Since $\langle g_i^{\perp_q}(x)\rangle$, $i=1$ or $2$, represents Hermitian dual codes of cyclic codes  $\langle g_i(x)\rangle$, so the above formula can be simplified to
	$$ \begin{aligned}
	=\langle[a(x)t(x)g_{1}(x)],[d(x)g_{2}^{\perp_q}(x)]\rangle_{h}
	+\langle[b(x) g_{2}(x)],[-c(x)\bar{t}^q(x)g_{1}^{\perp_q}(x)]\rangle_{h}\\
	+\langle[a(x)g_{1}(x)],[-d(x)\bar{t}^q(x)g_{2}^{\perp_q}(x)]\rangle_{h}
	+\langle[b(x) t(x) g_{2}(x)],[c(x)g_{1}^{\perp_q}(x)]\rangle_{h}. \\
	\end{aligned}$$
	
	By Lemma \ref{exchange_law}, we have that this equation is equal to zero, so $\mathcal{C}_0$ is contained by $\mathcal{C}^{\perp_q}_{q^2}(g_1,g_2,t)$.
	In addition, it is analogous that $\mathcal{C}^{\perp_q}_{q^2}(g_1,g_2,t)$ and $\mathcal{C}_0$ have the same dimension. 
	Therefore,  $\mathcal{C}_{0}$  is Hermitian dual code of $\mathcal{C}$.
	\end{proof}
	
	In fact, there are numerous $t(x)$ that satisfy the condition in Proposition \ref{dual_code}, so it is a simple matter to construct $\mathcal{C}_{q^2}(g_1,g_2,t)$ and $\mathcal{C}_0$ with parameter $\left[2 n, 2 n-\operatorname{deg}\left(g_{1}(x)\right)-\operatorname{deg}\left(g_{2}(x)\right)\right]$ and 
	$\left[2 n, \operatorname{deg}\left(g_{1}(x)\right)+\operatorname{deg}\left(g_{2}(x)\right)\right]$, respectively, during calculation by Magma \cite{bosma1997magma}.

	\begin{lemma}(\cite{lv2020quantum}, Proposition 2)\label{onegenearator}
	Suppose $f(x),g(x)$ are polynomials in $\mathcal{R}$ and $g(x)\mid x^n-1$. If $g^{\perp_q} (x) \mid g(x)$, then quasi-cyclic code generated by $(g(x),f (x)g(x))$ will be Hermitian self-orthogonal.
	\end{lemma}

	\begin{proposition}\label{dua_containing} If $g_{1}(x) \mid g_{1}^{\perp_q}(x)$, $g_{2}(x) \mid g_{2}^{\perp_q}(x)$, $g_{2}(x)\mid (t(x)+\bar{t}^q(x)) g_{1}^{\perp_q}(x)$ and $t(x)$ satisfy Proposition \ref{dual_code}, then $\mathcal{C}_{0}$ and $\mathcal{C}_{q^2}(g_1,g_2,t)$ are Hermitian self-orthogonal and Hermitian dual-containing, respectively.
	\end{proposition}
	\begin{proof}
	According to Lemma \ref{onegenearator}, if $g_{1}(x) \mid g_{1}^{\perp_q}(x)$, $g_{2}(x) \mid g_{2}^{\perp_q}(x)$, one can deduce that 1-generator quasi-cyclic codes generated by $( [-\bar{t}^q(x) g_{1}^{\perp_q}(x)], [g_{1}^{\perp_q}(x)])$  and  $( [g_{2}^{\perp_q}(x)],[- \bar{t}^q(x) g_{2}^{\perp_q}(x)])$ are both Hermitian self-orthogonal. 
	Let $c_1=( [-c(x)\bar{t}^q(x) g_{1}^{\perp_q}(x)], [c(x)g_{1}^{\perp_q}(x)])$  and  
	$c_2 =( [d(x)g_{2}^{\perp_q}(x)],[- d(x)\bar{t}^q(x) g_{2}^{\perp_q}(x)])$, $c(x)$, $d(x)$ $\in \mathcal{R}$, denote codewords generated by them, respectively. 
	Then, the Hermitian inner product of $c_1$ and $c_2$ can be expressed as follows.
	$$
	\begin{array}{l} 
	\langle  c_1,c_2 \rangle_h =\langle  [-c(x)\bar{t}^q(x) g_{1}^{\perp_q}(x)],[d(x)g_{2}^{\perp_q}(x)] \rangle_h+\langle  [c(x)g_{1}^{\perp_q}(x)],[- d(x)\bar{t}^q(x) g_{2}^{\perp_q}(x)] \rangle_h 	\\
	=\langle  [-c(x)\bar{t}^q(x) g_{1}^{\perp_q}(x)],[d(x)g_{2}^{\perp_q}(x)] \rangle_h+\langle  [c(x)t(x)g_{1}^{\perp_q}(x)],[- d(x) g_{2}^{\perp_q}(x)] \rangle_h 	\\
	=-\langle  [c(x)(t(x)+\bar{t}^q(x)) g_{1}^{\perp_q}(x)],[d(x)g_{2}^{\perp_q}(x)] \rangle_h.
	\end{array} 
	$$
	
	By Theorem \ref{inner_product_zero}, one can deduce that if $g_{2}(x)\mid (t(x)+\bar{t}^q(x)) g_{1}^{\perp_q}(x)$, then above equation is equal to 0, that is, $\mathcal{C}_{0}$ is Hermitian self-orthogonal. 
	Moreover, since $t(x)$ satisfies Proposition \ref{dual_code}, so $\mathcal{C}_{q^2}(g_1,g_2,t)$ is Hermitian dual-containing.
	\end{proof}

	
	\begin{theorem}\label{method} Let  $\mathcal{C}_{q^2}(g_1,g_2,t)$  be a quasi-cyclic code proposed in Definition  \ref{def1}.  If  $g_{1}(x)$, $g_{2}(x)$ and $t(x)$ satisfies Proposition \ref{dual_code} and \ref{dua_containing}, then there exists a pure  $\left[\left[2n, 2n-2\operatorname{deg}\left(g_{1}(x)\right)-2\operatorname{deg}\left(g_{2}(x)\right), d\right]\right]_{q}$  QECC, where  $d=\min \left\{wt(\vec{c}) \mid \vec{c} \in \mathcal{C}_{q^2}(g_1,g_2,t)\right\}$.\end{theorem}

	\begin{remark}
	Let  $T_{1}$ and  $T_{2}$  be defining sets of cyclic  $\operatorname{codes}\left\langle g_{1}(x)\right\rangle$  and  $\left\langle g_{2}(x)\right\rangle$, respectively. If $T_{1} \cap T^{-q}_{1}=\emptyset$ and $T_{2} \cap T^{-q}_{2}=\emptyset$, then  $g_{1}(x) \mid g_{1}^{\perp_q}(x)$  and $g_{2}(x) \mid g_{2}^{\perp_q}(x)$.
	\end{remark}
	
	Let $\omega$ be a primitive element of $F_4$. For simplicity, elements $0$, $1$, $\omega$, $\omega^2$ in $F_4$ are represented by $0$, $1$, $2$, $3$, respectively.
	To save space, we express coefficient polynomials in ascending order and use indexes of elements to express the same number of consecutive elements. For example, polynomial  
	$1+x^3+x^5$ can be written as $10^2101$.
	
	\begin{example}
	Let  $q^2=4$  and  $n=41$. Consider the $4$-cyclotomic cosets modulo $41$. Select  $T_{1}=C_{1} $ and  $T_{2}=C_{3}$  as the defining sets of cyclic codes  $\left\langle g_{1}(x)\right\rangle$  and  $\left\langle g_{2}(x)\right\rangle$. Then  $g_1(x)=10320102301$, $g_2(x)=12^{3}1312^{3}1$. We choose  $t(x)=10203^{5}2130^{2}2^{2}3^{2}102^{2}3010^{2}1313^{2}2031^{2}3032$.  This will generate a $[82, 62, 9]_4$ code, one can verify which is a dual-containing code with respect to Hermitian inner product, whose weight distribution can be written as  $w(z)=1+5166 z^{9}+119310 z^{10}+2263323 z^{11}+\cdots+ 1209882125048724140018184234 z^{82}$.
	Then a new binary QECC with parameters  $[[ 82, 42, 9 ]]_{2}$  can be provided. Observe that a code with parameter  $[[82, 42, 8]]_{2}$  is the best-known binary QECC with length $82$ and dimension $42$ in \cite{Grassl:codetables}, so the current record of corresponding minimum distance can be improved to $9$. 
	\end{example}
	
	We also obtain a new binary QECCs with parameters $[[70,48,6]]_{2}$ from Theorem \ref{method}, which is better than  corresponding QECC  $[[70,48,5]]_{2}$ that appeared in Grassl's code tables \cite{Grassl:codetables}. Select $g_1(x)=1^{2}3120^{3}21$,
	$g_2(x)=131$, and
	$t(x)=101212010312^{3}30$ $2013021310^{2}20303^{3}$. A Hermitian dual-containing $[70,59,6]_4$ code can be obtained, whose weight distribution is $w(z)=1+ 21945 z^{6}+638400 z^{7}+14751240 z^{8}+\cdots+ 596798778770743310728680969 z^{70}$.\par
	
	Moreover, Theorem \ref{method} is a rigorous method for constructing QECCs but is not a necessity. For $g_1(x)$ and $g_2(x)$ that do not satisfy the condition Proposition \ref{dua_containing}, QECCs with good parameters can also be constructed by searching for suitable $t(x)$. We identify conditions that $t(x)$ should satisfy, which help to expand the choice of $g_1(x)$ and $g_2(x)$ and thus make it more likely to construct new quantum codes.

	\begin{theorem}\label{onegenearator_sufficient}
	Suppose $\mathcal{C}$ is a quasi-cyclic code generated by $([f(x)g(x)],[g(x)])$, $f(x),g(x) \in \mathcal{R}$, then the sufficient and necessary condition for $\mathcal{C}$ to be Hermitian self-orthogonal is $g^{\perp_q} (x) \mid (f(x)\bar{f_q}(x)+1)g(x)$.
	\end{theorem}
	\begin{proof}
	Let $c=([a(x)f(x)g(x)],[a(x)g(x)])$, $a(x)\in \mathcal{R}$. Then any codeword of $\mathcal{C}$ can be represented by $c$. The Hermitian inner product of $c$ can be written as $$\begin{array}{l} \langle c,c \rangle_h 
		= \langle [a(x)f(x)g(x)],[a(x)f(x)g(x)] \rangle_h+\langle[a(x)g(x)],[a(x)g(x)] \rangle_h\\
		= \langle [a(x)f(x)\bar{f_q}(x)g(x)],[a(x)g(x)] \rangle_h+\langle[a(x)g(x)],[a(x)g(x)] \rangle_h\\
		=\langle [a(x)(f(x)\bar{f_q}(x)+1)g(x)],[a(x)g(x)] \rangle_h \end{array}$$
	
	With reference to Theorem \ref{inner_product_zero}, the sufficient and necessary condition for above equation to be zero is $g^{\perp_q} (x) \mid (f(x)\bar{f_q}(x)+1)g(x)$. In addition, by Lemma \ref{Hermitianself-orthogonal}, it is easy to deduce that the sufficient and necessary condition for $\mathcal{C}$ to be Hermitian self-orthogonal is $g^{\perp_q} (x) \mid (f(x)\bar{f_q}(x)+1)g(x)$.
	\end{proof}

	\begin{theorem}\label{extend_condition}
	Let  $\mathcal{C}_{q^2}(g_1,g_2,t)$  be a quasi-cyclic code proposed in Definition  \ref{def1}.  
	If $g_{2}(x)\mid (t(x)+\bar{t}^q(x)) g_{1}^{\perp_q}(x)$, $g_i(x) \mid (t(x)\bar{t}^q(x)+1)g^{\perp_q}_i(x)$, $i=1$ or $2$, and $t(x)$ satisfies Proposition \ref{dual_code}, then there exists a pure  $\left[\left[2n, 2n-2\operatorname{deg}\left(g_{1}(x)\right)-2\operatorname{deg}\left(g_{2}(x)\right), d\right]\right]_{q}$  QECC, where  $d=\min \left\{wt(\vec{c}) \mid \vec{c} \in \mathcal{C}_{q^2}(g_1,g_2,t)\right\}$.
	\end{theorem}
	
	\begin{proof}
	By Proposition \ref{dua_containing} and Theorem \ref{onegenearator_sufficient}, we can deduce that $\mathcal{C}_0$ is Hermitian self-orthogonal. In addition, $t(x)$ satisfies Proposition \ref{dual_code}, so $\mathcal{C}_0$ is a $\left[2 n, \operatorname{deg}\left(g_{1}(x)\right)+\operatorname{deg}\left(g_{2}(x)\right)\right]_q$ Hermitian self-orthogonal code. By Lemma \ref{construction}, there exists a pure  $\left[\left[2n, 2n-2\operatorname{deg}\left(g_{1}(x)\right)-2\operatorname{deg}\left(g_{2}(x)\right), d\right]\right]_{q}$  QECC, where  $d=\min \left\{wt(\vec{c}) \mid \vec{c} \in \mathcal{C}_{q^2}(g_1,g_2,t)\right\}$. 
	So we complete the proof.
	\end{proof}
	
	Here we give Examples \ref{notsat}, which satisfy Theorem \ref{extend_condition}, but not Theorem \ref{method}.
	
	\begin{example}\label{notsat}
	Let  $q^2=4$  and  $n=35$. Consider the $4$-cyclotomic cosets modulo $35$. Select  $T_{1}=C_{0} \cup C_{1}$ and  $T_{2}=C_{5}\cup C_{7}$  as the defining sets of cyclic codes  $\left\langle g_{1}(x)\right\rangle$  and  $\left\langle g_{2}(x)\right\rangle$. Then  $g_1(x)=1^{2}3023^{2}1$,
	$g_2(x)=12^{2}031$. We choose  $t(x)=032^{2}0^{3}2^{4}012^{2}31^{2}03^{2}21^{2}20203^{2}1032$.  
	It is easy to check that $g^{\perp_q}_1(x)\mid g_1(x)$, $g_{2}(x)\mid (t(x)+\bar{t}^q(x)) g_{1}^{\perp_q}(x)$, and $g_i(x) \mid (t(x)\bar{t}^q(x)+1)g^{\perp_q}_i(x)$, $i=1$ or $2$. In addition, one can calculate that parameters of $\mathcal{C}_{g_1,g_2,t}$ and $\mathcal{C}_0$ are satisfy Proposition \ref{dual_code}.
	Therefore, this will generate a Hermitian dual-containing $[70, 56, 7]_4$ code, whose weight distribution can be written as  $w(z)=1+10605 z^{7}+242025 z^{8}+4743585 z^{9}+\cdots+ 9324980918306331190370055 z^{70}$.
	Then a new binary QECC with parameters  $[[ 70, 42, 7 ]]_{2}$  can be provided. Observe that a code with parameter  $[[70,42,6]]_{2}$  is the best-known binary QECC with length $70$ and dimension $42$ in \cite{Grassl:codetables}, so the current record of corresponding minimum distance can be improved to $7$. 
	\end{example}

	\begin{remark} Two new binary QECCs with parameters $[[ 42, 14, 8 ]]_{2}$, $[[170,148,5]]_{2}$ derived from Theorem \ref{extend_condition} are also given here, both of which are better than the corresponding QECCs that appeared in Grassl's code tables \cite{Grassl:codetables}. Their generators and weight distributions are as follows.\par
	(1)	 $[[42, 14, 8]]_{2}$: $g_1(x)=1^{3}01^{2}0^{2}1^{2}$, $g_2(x)=23231^{2}$, and
	$t(x)=2010^{2}$ $10201^{2}210120^{2}1^{2}3$. A Hermitian dual-containing $[42,28,8]_4$ code can be obtained, whose weight distribution is $w(z)=1+ 4662 z^{8}+40152 z^{9}+345933 z^{10}+\cdots+407672056605 z^{42}$.\par
	(2)	 $[[ 170,148,5 ]]_{2}$: $g_1(x)=1^{3}0^{3}12^{2}31$,
	$g_2(x)=1^{2}$, and
	$t(x)=20123^{2}01$ $3^{5}020^{2}1^{2}0^{2}102303^{2}02^{2}1012102^{2}3^{4}131^{2}0^{3}323123^{2}212131013012^{2}0^{2}1^{2}3^{2}230120302$. A Hermitian dual-containing $[170, 159, 5] _4$ code can be obtained, whose weight distribution is  $w(z)=1+66045 z^{5}+5326950 z^{6}+373937100 z^{7}+\cdots+ 
	3075766749787575402851$ $57212220447157605514941639416619352831833745300990444 z^{170}$.	 	\\
	\end{remark}
	
	Furthermore, by the propagation rules of QECCs in Lemma \ref{propagation_rule}, we can get another 11 new binary QECCs from QECCs above. As shown in Table \ref{tab: binary}, their parameters also improve the lower bounds on the minimum distance in Grassl's table \cite{Grassl:codetables}.
	
	\begin{table}[h]
	\caption{New binary quantum codes from propagation rules}
	\label{tab: binary}
	\begin{center}

		\begin{tabular}{ccc}
			\hline
			NO. &    Our Codes     & Codes in \cite{Grassl:codetables} \\ \hline
	 1  & $[[42,13,8]]_2$  &      $  [[42,13,7]]_2$      \\
	 2  & $[[71,48,6]]_2$  &       $[[71,48,5]]_2$       \\
	 3  & $[[70,39,7]]_{2}$ &      $[[70,39,6]]_2$       \\
	 4  & $[[70,40,7]]_{2}$ &      $[[70,40,6]]_2$       \\ 
	 5  & $[[70,41,7]]_{2}$ &      $[[70,41,6]]_2$       \\
	 6  & $[[71,40,7]]_{2}$ &      $[[71,40,6]]_2$       \\ 
	 7  & $[[71,41,7]]_{2}$ &      $[[71,41,6]]_2$       \\
	 8  & $[[71,42,7]]_{2}$ &      $[[71,42,6]]_2$       \\ 
	 9  & $[[83, 42, 9]]_{2}$ &      $[[83,42,8]]_2$       \\
	 10  & $[[170,147,5]]_2$ &      $[[170,147,4]]_2$       \\
	 11  & $[[171,148,5]]_2$ &      $[[171,148,4]]_2$       \\
	\hline
		\end{tabular}\end{center}
	\end{table}
	
	To demonstrate the effectiveness of our approach, we also constructed many QECCs over $F_q$, where $q$ is $3,4,5$. 
	In particular, they all beat the quantum GV bound and are new or have better parameters than those in the literature. Let  $\gamma$, $\xi$, and  $\zeta$  be the primitives over  $F_{9}$, $F_{16}$, and  $F_{25}$, respectively. Elements  $0$, $1$, $\gamma$, $\cdots$, $\gamma ^{7}$ in $F_{9}$; $0$, $1$, $\xi$, $\cdots$, $\xi^{14}$ in  $F_{16}$  and  $0$, $1$, $\zeta$, $\cdots$, $\zeta^{23}$ in  $F_{25}$  are represented by  $0$, $1$, $\cdots$, $8$; $0$, $1$, $\cdots$, $9$, $A$, $\cdots$, $F$  and  $0$, $1$, $\cdots$, $9$, $A$, $\cdots$, $O$, respectively. In Tables \ref{codesF9}, \ref{codesF16} and \ref{codesF25} are the quasi-cyclic codes over $F_9$, $F_{16}$ and $F_{25}$, which are dual-containing under the Hermitian inner product. These codes are used to construct QECCs in Tables \ref{QECC_F3}, \ref{QECC_F4}, and \ref{QECC_F5}.

	\begin{table}[h]
	\begin{center}
	\caption{Dual-containing quasi-cyclic codes $\mathcal{C}_9(g_1,g_2,t)$}\label{codesF9}%
	\begin{tabular}{@{}lll@{}}
	\toprule
	$n$      & $g_1(x),g_2(x),t(x)$                                                       & $\mathcal{C}_9(g_1,g_2,t)$ \\ \midrule
	10       & $121$, $51$, $41781$                                                       & $[10,7,4]_9$     \\
	16       & $3241$, $721$, $07^{2}3568$                                                & $[16,11,5]_9$    \\
	20       & $587341$, $13031$, $20^{4}21^{4}$                                          & $[20,11,7]_9$    \\
	22       & $50151^{2}$, $5^{2}1501$, $35^{5}1^{5}$                                    & $[22,12,8]_9$    \\
	26       & $1^{2}0^{2}1$, $501^{2}$, $6431^{10}$                                      & $[26,19,6]_9$    \\
	28       & $5641$, $1725^{2}15^{2}431$, $12053867462858$                              & $[28, 15, 10]_9$ \\
	32       & $20701$, $7157^{2}1$, $37404^{2}3604036434$                                & $[32, 21, 8]_9$  \\
	32       & $2651$, $7157^{2}1$, $5320341^{10}$                                        & $[32,22,7]_9$    \\
	32       & $2651$, $6871$, $7561^{13}$                                                & $[32,24,6]_9$    \\
	32       & $2651$, $61$, $6^{3}21^{12}$                                               & $[32,25,5]_9$    \\
	34       & $146424641$, $51$, $23630471872626364$                                     & $[34,16,7]_9$    \\
	40       & $1361$, $51^{2}71$, $82326108738640656048$                                 & $[40, 33, 5]_9$  \\
	44       & $515^{2}15^{3}1^{3}$, $51$, $0^{2}1^{3}5621^{2}710503821473$               & $[44,33,7]_9$    \\
	44       & $15^{3}01$, $50151^{2}$, $37250248^{2}71^{2}81^{2}2036383$                 & $[44,34,7]_9$    \\
	46       & $50^{2}10105^{2}1^{3}$, $51$, $343036^{2}405875483242425$                  & $[46,34,8]_9$    \\
	52       & $1515^{3}1$, $1501$, $54127307140531642^{2}81070381$                       & $[52,40,7]_9$    \\
	70       & $1^{2}27031$, $16581$, $7874818^{2}136761517801545^{2}45167207674$         & $[70,58,7]_9$    \\
	70       & $1^{2}27031$, $5821$, $7465761^{2}78242^{2}874070287^{2}6252^{2}143^{2}02$ & $[70,59,6]_9$    \\
	\bottomrule 
	\end{tabular}
	\end{center}
	\end{table}
	
	\begin{table}[h]
	\begin{center}
		\caption{Dual-containing quasi-cyclic codes $\mathcal{C}_{16}(g_1,g_2,t)$}\label{codesF16}%
		\begin{tabular}{@{}lll@{}}
			\toprule
			$n$      & $g_1(x),g_2(x),t(x)$                                                               & $\mathcal{C}_9(g_1,g_2,t)$ \\ \midrule
			10       & $ A 3 1$, $ 1^{2}$, $ 9 B 9 5 8$                                                   & $[10,7,4]_{16}$  \\
			14       & $ 1^{2} 0 1$, $ 1 0 1^{2}$, $ E 2 B 2 9 1 9$                                       & $[14,8,6]_{16}$  \\
			18       & $ 6 0^{2} 1$, $ 6 B 0 B 1$, $ B 2 E B A E B A 2$                                   & $[18,11,6]_{16}$ \\
			22       & $ 1 6 1^{2} B 1$, $ 1^{2}$, $ 0 4 2 1 3 1 D 6 E 7 A$                               & $[22,16,6]_{16}$ \\
			26       & $ 1 5 4 B D 2 1$, $ 1 A 6 7 1$, $ B 8 B 3 5 D 3 2 D 6 2 B 6$                       & $[26,16,8]_{16}$ \\
			26       & $ 1 E 8 1$, $ 1 9 7 6 A 3 1$, $ D 7 C 7 5 F 9 1 A F 6 9 6$                         & $[26,17,7]_{16}$ \\
			30       & $ 7 1 5 D 1$, $ 1 5 F 1$, $ E 1 5^{3} 7 3 A F 7 C 1 5 8 5$                         & $[30,23,6]_{16}$ \\
			30       & $ 7 1 5 D 1$, $ A 9 1$, $ D 2 7 1 7 B 8 7^{2} 6 B E B 5 B$                         & $[30,24,5]_{16}$ \\
			38       & $ 1 6 0 6^{2} B^{2} 0 B 1$, $ 1^{2}$, $ 2 B E 6 F 6 9 0 9 3 D B 9^{2} E 0 4 D^{2}$ & $[38,28,7]_{16}$ \\
			40       & $ 1 0 1 0 1^{3}$, $ 6 1 6 B^{2} 1$, $ E 1^{2} 2 F 7 4 7 6 E 3 8 3 8 C A B 2 3 6 A$ & $[42,31,7]_{16}$ \\
			\bottomrule 
		\end{tabular}
	\end{center}
	\end{table}
	
	\begin{table}[h]
	\begin{center}
		\caption{Dual-containing quasi-cyclic codes $\mathcal{C}_{25}(g_1,g_2,t)$}\label{codesF25}%
		\begin{tabular}{@{}lll@{}}
			\toprule
			$n$      & $g_1(x),g_2(x),t(x)$                  & $\mathcal{C}_9(g_1,g_2,t)$ \\ \midrule
			8        & $JD1$, $D1$, $8AK^{2}$                & $[8,5,4]_{25}$   \\
			12       & $DG1$, $D1$, $LF8BHJ$                 & $[12,9,4]_{25}$  \\
			14       & $DE61$, $1EDI1$, $4D8^{2}21^{2}$      & $[14,7,7]_{25}$  \\
			14       & $DE61$, $DI21$, $141^{5}$             & $[14,8,6]_{25}$  \\
			16       & $DA0A1$, $1AJM1$, $E1^{7}$            & $[16,8,8]_{25}$  \\
			16       & $J0D01$, $ADM1$, $L6M1^{5}$           & $[16,9,7]_{25}$  \\
			16       & $4141$, $ADM1$, $E1^{7}$              & $[16,10,6]_{25}$ \\
			16       & $1G1$, $7DJ1$, $721^{6}$              & $[16,11,5]_{25}$ \\
			22       & $DJ1D^{2}1$, $D1$, $MA21^{8}$         & $[22,16,6]_{25}$ \\
			26       & $1OEO1$, $DH51$, $JGBDCA4O7B5I5$      & $[26,19,6]_{25}$ \\
			26       & $DH51$, $161$, $KG3N6CKGHA4K$         & $[26,21,5]_{25}$ \\
			32       & $DMG1$, $M4J1$, $AD2^{2}MF9M12GD2EN2$ & $[32,26,5]_{25}$ \\
			\bottomrule
		\end{tabular}
	\end{center}
	\end{table}
	
	\begin{table}[h]
	\begin{center}
	\caption{New QECCs over $F_3$}\label{QECC_F3}%
	\begin{tabular}{@{}ccccc}
	\toprule
	$\mathcal{C}_9(g_1,g_2,t)$ &      QECCs       & Codes in \cite{lv2020explicit} & Codes in \cite{ezerman2019good} & Codes in \cite{Edel2022} \\ \midrule
	$[10,7,4]_9$   & $[[10,4,4]]_3$  &           -            &              -              &        $[[11,1,4]]_3$         \\
	$[16,11,5]_9$   & $[[16,6,5]]_3$  &     $[[16,5,5]]_3$     &              -              &               -               \\
	$[20,11,7]_9$   & $[[20,2,7]]_3$  &           -            &              -              &        $[[23,1,5]]_3$         \\
	$[22,12,8]_9$   & $[[22,2,8]]_3$  &           -            &              -              &        $[[23,1,5]]_3$         \\
	$[26,19,6]_9$   & $[[26,12,6]]_3$ &           -            &              -              &        $[[26,11,6]]_3$        \\
	$[28,15,10]_9$  & $[[28,2,10]]_3$ &           -            &              -              &               -               \\
	$[32,21,8]_9$   & $[[32,10,8]]_3$ &           -            &              -              &               -               \\
	$[32,22,7]_9$   & $[[32,12,7]]_3$ &           -            &              -              &               -               \\
	$[32,24,6]_9$   & $[[32,16,6]]_3$ &           -            &       $[[33,15,6]]_3$       &               -               \\
	$[32,25,5]_9$   & $[[32,18,5]]_3$ &           -            &              -              &               -               \\
	$[34,16,7]_9$   & $[[34,16,7]]_3$ &           -            &              -              &               -               \\
	$[40,33,5]_9$   & $[[40,26,5]]_3$ &           -            &              -              &        $[[40,24,5]]_3$        \\
	$[44,33,7]_9$   & $[[44,22,7]]_3$ &           -            &              -              &        $[[44,4,7]]_3$         \\
	$[44,34,7]_9$   & $[[44,24,7]]_3$ &           -            &              -              &        $[[44,4,7]]_3$         \\
	$[46,34,8]_9$   & $[[46,22,8]]_3$ &           -            &              -              &        $[[45,3,8]]_3$         \\
	$[52,40,7]_9$   & $[[52,28,7]]_3$ &           -            &              -              &        $[[52,28,6]]_3$        \\
	$[70,58,7]_9$   & $[[70,46,7]]_3$ &           -            &              -              &        $[[71,45,4]]_3$        \\
	$[70,59,6]_9$   & $[[70,48,6]]_3$ &           -            &              -              &        $[[73,37,6]]_3$        \\
	\bottomrule  
	\end{tabular}
	\end{center}
	\end{table}
	
	\begin{table}[h]
	\begin{center}
		\caption{New QECCs over $F_4$}\label{QECC_F4}%
		\begin{tabular}{@{}cccc@{}}
			\toprule
			$\mathcal{C}_{16}(g_1,g_2,t)$ &      QECCs       & Codes in \cite{Lv2019newnon} & \cite{Edel2022} \\ \midrule
	$[10,7,4]_{16}$   & $[[10,4,4]]_4$  &             -             &    $[[9,1,4]]_4$     \\
	$[14,8,6]_{16}$   & $[[14,2,6]]_4$  &             -             &          -           \\
	 $[18,11,6]_{16}$   & $[[18,4,6]]_4$  &             -             &    $[[18,4,5]]_4$    \\
	 $[22,16,6]_{16}$   & $[[22,10,6]]_4$ &             -             &    $[[22,8,6]]_4$    \\
	 $[26,16,8]_{16}$   & $[[26,6,8]]_4$  &             -             &          -           \\
	 $[26,17,7]_{16}$   & $[[26,8,7]]_4$  &             -             &          -           \\
	 $[30,23,6]_{16}$   & $[[30,16,6]]_4$ &      $[[31,16,5]]_4$      &          -           \\
	 $[30,24,5]_{16}$   & $[[30,18,5]]_4$ &      $[[33,17,5]]_4$      &          -           \\
	 $[38,28,7]_{16}$   & $[[38,18,7]]_4$ &             -             &   $[[39,21,5]]_4$    \\
	 \bottomrule 
		\end{tabular}
	\end{center}
	\end{table}
	

	\begin{table}[h]
	\begin{center}
		\caption{New QECCs over $F_5$}\label{QECC_F5}%
		\begin{tabular}{@{}cccc@{}}
			\toprule
			$\mathcal{C}_{25}(g_1,g_2,t)$ &       QECCs       & Codes in \cite{yao2021new}& \cite{Edel2022} \\ \midrule
	   $[8,5,4]_{25}$         &  $[[8,2,4]]_{5}$  &             -              &    $[[7,1,3]]_5$     \\
	   $[12,9,4]_{25}$        & $[[12,6,4]]_{5}$  &             -              &    $[[16,6,4]]5$     \\
	   $[14,7,7]_{25}$        & $[[14,0,7]]_{5}$  &             -              &          -           \\
	   $[14,8,6]_{25}$        & $[[14,2,6]]_{5}$  &             -              &    $[[19,1,5]]_5$    \\
	   $[16,8,8]_{25}$        & $[[16,0,8]]_{5}$  &             -              &          -           \\
	   $[16,9,7]_{25}$        & $[[16,2,7]]_{5}$  &             -              &    $[[19,1,5]]_5$    \\
	  $[16,10,6]_{25}$        & $[[16,4,6]]_{5}$  &             -              &          -           \\
	  $[16,11,5]_{25}$        & $[[16,6,5]]_{5}$  &             -              &    $[[16,6,4]]_5$    \\
	  $[22,16,6]_{25}$        & $[[22,10,6]]_{5}$ &             -              &          -           \\
	  $[26,19,6]_{25}$        & $[[26,12,6]]_{5}$ &      $[[26,10,6]]_5 $      &          -           \\
	  $[26,21,5]_{25}$        & $[[26,16,5]]_{5}$ &             -              &   $[[27,13,5]]_5$    \\
	  $[32,26,5]_{25}$        & $[[32,26,5]]_{5}$ &             -              &   $[[32,21,5]]_5$    \\
	      \bottomrule  
		\end{tabular}
	\end{center}
	\end{table}
	
	However, even for supercomputers, it is very difficult to calculate the specific parameters of quasi-cyclic codes of larger length and dimension. Therefore, we give a lower bound on the minimum distance of $\mathcal{C}(g_1,g_2,t)$, which will help to reduce the computational effort.
	
	\begin{theorem}\label{lowerbound}
	The lower bound on the minimum distance of  $\mathcal{C}_{q^2}(g_1,g_2,t)$ is
	$$
	d(\mathcal{C}_{q^2}(g_1,g_2,t))\ge \min \left\{\begin{array}{l}
	d(lcm(g_2(x),\frac{{{x^n} - 1}}{{gcd({x^n} - 1,t(x))}})) \\
	d(lcm(g_1(x),\frac{{{x^n} - 1}}{{gcd({x^n} - 1,t(x))}})) \\
	d(g_{1}(x))+d(g_{1}(x) t(x)) \\
	d(g_{2}(x))+d(g_{2}(x) t(x)) \\
	2 d(g c d(g_{1}(x) t(x), g_{2}(x))) \\
	d(g c d(l c m(g_{1}(x), \frac{g_{2}(x)}{g c d(g_{2}(x), t(x))}), l c m(t(x) g_{2}(x), g_{1}(x) t^{2}(x)))) \\
	d(g c d(l c m(g_{2}(x), \frac{g_{1}(x)}{g c d(g_{1}(x), t(x))}), l c m(g_{2}(x) t^{2}(x), t(x) g_{1}(x)))) .
	\end{array}\right.$$
	\end{theorem}
	\begin{proof}
	Let $a(x),b(x)$ are polynomials in $\mathcal{R}$. Then any codeword of $\mathcal{C}_{q^2}(g_1,g_2,t)$ can be represented by $c=([a(x) t(x) g_{1}(x)+b(x) g_{2}(x)],[a(x)g_{1}(x)+b(x) t(x) g_{2}(x)])$.
	
	\textbf{Case (a):} 
	If $a(x)=0$ and $b(x)\ne 0$, then $c=([b(x) g_{2}(x)],[b(x) t(x) g_{2}(x)])$.
	
	\textbf{(i):}
	If $b(x) t(x) g_{2}(x)\equiv 0\pmod{x^n-1}$, i.e., $x^n-1\mid b(x)t(x)g_2(x)$, then $\frac{x^n-1}{gcd(x^n-1,t(x))} \mid b(x)g_2(x)$. So $\left \langle b(x)g_2(x) \right \rangle$ is a subcode of $ \langle \frac{x^n-1}{gcd(x^n-1,t(x))} \rangle$, there will $d(\mathcal{C}_{q^2}(g_1,g_2,t))=d(b(x)g_2(x))\ge d(lcm(g_2(x),\frac{{{x^n} - 1}}{{gcd({x^n} - 1,t(x))}}))$.
	
	\textbf{(ii):}
	If $b(x) t(x) g_{2}(x)\ne 0\pmod{x^n-1}$, there will $d(\mathcal{C}_{q^2}(g_1,g_2,t))\ge d(g_2(x))+d(t(x)g_2(x))$.
	
	\textbf{Case (b):} 
	If $a(x)\ne 0$ and $b(x)= 0$, then $c=([a(x) t(x)g_{1}(x)],[a(x)g_{1}(x)])$. Similar to Case a, if $b(x)  t(x)g_{2}(x)\equiv 0\pmod{x^n-1}$, then 	$d(\mathcal{C}_{q^2}(g_1,g_2,t))\ge d(lcm(g_1(x),\frac{{{x^n} - 1}}{{gcd({x^n} - 1,t(x))}}))$. Otherwise, $d(\mathcal{C}_{q^2}(g_1,g_2,t))\ge d(g_1(x))+d(t(x)g_1(x))$.
	
	\textbf{Case (c):} 
	If $a(x)\ne 0$ and $b(x)\ne0$, $c=([a(x) t(x) g_{1}(x)+b(x) g_{2}(x)],[a(x)g_{1}(x)+b(x) t(x) g_{2}(x)])$.
	
	\textbf{(i):} 
	If $a(x) t(x) g_{1}(x)+b(x) g_{2}(x)\ne0 \pmod{x^n-1}$, and $a(x)g_{1}(x)+b(x) t(x) g_{2}(x)\ne0\pmod{x^n-1}$, then $\left \langle a(x)t(x)g_1(x) \right \rangle$ and 
	$\left \langle b(x)g_2(x) \right \rangle$ are subcodes of $\left \langle gcd(t(x)g_1(x),g_2(x)) \right \rangle$, and $\left \langle a(x)g_1(x) \right \rangle$ and 
	$\left \langle b(x)t(x)g_2(x) \right \rangle$ are subcodes of $\left \langle gcd(g_1(x),t(x)g_2(x)) \right \rangle$. So $d(\mathcal{C}_{q^2}(g_1,g_2,t))\ge d(gcd(t(x)g_1(x),g_2(x)))+d(gcd(g_1(x),t(x)g_2(x)))=2d(gcd(g_1(x),$ $t(x)g_2(x)))$.
	
	\textbf{(ii):} 
	If $a(x) t(x) g_{1}(x)+b(x) g_{2}(x)\equiv0 \pmod{x^n-1}$, $a(x)g_{1}(x)+b(x)t(x) g_{2}(x)\ne0\pmod{x^n-1}$, then $x^n-1\mid a(x) t(x) g_{1}(x)+b(x) g_{2}(x)$. In quotient ring  $\mathcal{R}$, there will $a(x)t(x)g_1(x)\equiv-b(x)g_2(x)\pmod{x^n-1}$, so $g_2(x)\mid a(x)t(x)g_1(x)$ and $\frac{g_2(x)}{gcd(g_2(x),t(x))}\mid \frac{a(x)t(x)g_1(x)}{gcd(g_2(x),t(x))}$.
	Since $gcd (\frac{g_2(x)}{gcd(g_2(x),t(x))}, \frac{t(x)}{gcd(g_2(x),t(x))})=1$,
	there is $\frac{g_2(x)}{gcd(g_2(x),t(x))}\mid a(x)g_1(x)$, i.e., $\left \langle a(x)g_1(x) \right \rangle$ is subcode of $ \langle \frac{g_2(x)}{gcd(g_2(x),t(x))}  \rangle$ and $d(a(x)g_1(x))\ge d(lcm(g_1(x),\frac{g_2(x)}{gcd(g_2(x),t(x))}))$.
	
	Similarly, it follows that $g_1(x)t^2(x)\mid b(x)t(x)g_2(x)$, i.e., $\langle b(x)t(x)g_2(x) \rangle$ is contained by $ \langle g_1(x)t^2(x) \rangle$ and $d(b(x)t(x)g_2(x))\ge d(lcm(t(x)g_2(x),g_1(x)t^2(x)))$.
	So $d(\mathcal{C}_{q^2}(g_1,g_2,t))\ge d(gcd( lcm(g_1(x),\frac{g_2(x)}{gcd(g_2(x),t(x))}),lcm(t(x)g_2(x),g_1(x)t^2(x)))) $.
	
	\textbf{(iii):} 
	If $a(x) t(x) g_{1}(x)+b(x) g_{2}(x)\ne0\pmod{x^n-1}$, and $a(x)g_{1}(x)+b(x) t(x) g_{2}(x)\equiv0\pmod{x^n-1}$, then $x^n-1\mid a(x)  g_{1}(x)+b(x) t(x)g_{2}(x)$.
	So, in $\mathcal{R}$, one can easy to judge that $-a(x)g_1(x)\equiv b(x)t(x)g_2(x)\pmod{x^n-1}$, which yields $g_1(x)\mid b(x)t(x)g_2(x)$ and  $\frac{g_1(x)}{gcd(g_1(x),t(x))}\mid \frac{b(x)t(x)g_2(x)}{gcd(g_1(x),t(x))}$. Since $gcd(\frac{g_1(x)}{gcd(g_1(x),t(x))}, \frac{t(x)}{gcd(g_1(x),t(x))})=1$, there is $\frac{g_1(x)}{gcd(g_1(x),t(x))}\mid b(x)g_2(x)$
	and $d(b(x)g_2(x))\ge d(lcm(g_2(x),\frac{g_1(x)}{gcd(g_1(x),t(x))}))$.
	Moreover, as $g_2(x)t^2(x)\mid a(x)t(x)g_1(x)$ and $d(b(x)t(x)g_2(x))\ge d(lcm(g_2(x)t^2(x),t(x)g_1(x)))$, so we have $d(\mathcal{C}_{q^2}(g_1,g_2,t))\ge d (gcd(lcm(g_2(x),\frac{g_1(x)}{gcd(g_1(x),t(x))}), lcm(g_2(x)t^2(x),t(x)g_1(x)))) $.
	
	In summary, the lower bound on the minimum distance of $\mathcal{C}_{q^2}(g_1,g_2,t)$ is proved.
	\end{proof}
	
	Here, we provide two examples to show that the lower bound in Theorem \ref{lowerbound} is helpful and feasible to deduce good QECCs.
	\begin{example} 
	Let  $q^2=4$  and  $n=133$. Consider the $4$-cyclotomic cosets modulo $133$. 
	Select  $T_{1}=C_{1}\cup C_{2} $ and  $T_{2}=C_{1}\cup C_{19}$  as the defining sets of cyclic codes  $\left\langle g_{1}(x)\right\rangle$  and  $\left\langle g_{2}(x)\right\rangle$. 
	So $g_1(x)=1010101^{2}0^{3}1^{3}01^{2}01$,
	$g_2(x)=10^{4}1231301^{2}$. We choose  $t(x)=31^{2}23^{2}213^{5}20312^{2}132^{2}13101030^{2}1$ $32120^{2}2121023^{2}202321^{2}20131^{2}0213210^{4}20303103213102023^{2}030303^{2}21^{2}020^{3}31^{2}202$ $10^{2}120^{2}3^{2}2130^{2}320^{3}12^{2}03$.
	This will generate a $[ 266, 236]_4$ code, which can be verified as a dual-containing code concerning Hermitian inner product.
	Moreover, we have calculated that lower bound on the minimum distance of  this code is 6, so there exists a binary QECC with parameters  $[[ 266, 206, \ge6 ]]_{2}$,  which is better than the known $[[267,201,6]]_2$ QECC appeared in \cite{Edel2022}.
	\end{example}
	
	\begin{example} 
	Let  $q^2=9$  and  $n=247$. Consider the $4$-cyclotomic cosets modulo $247$. 
	Select  $T_{1}=T_{2}=C_{1}\cup C_{38}$ as the defining sets of cyclic codes  $\left\langle g_{1}(x)\right\rangle$  and  $\left\langle g_{2}(x)\right\rangle$. 
	Then $g_1(x)=g_2(x)=176853^{2}185651$. We choose $t(x)=4012063631^{2}568472542^{2}61^{2}871^{2}7825737536751^{2}38^{2}067638217537^{2}030406702^{2}14673
	$ $131^{2}43^{2}47^{2}3861^{2}0625012120343782832^{3}5232423201^{2}64368243701815823860^{2}342804^{2}$ $83^{2}7674^{2}186256^{2}5608717131470103047165835754305^{2}4053870^{2}26308750^{2}656^{2}86378$ $7^{2}53^{2}186515304035314^{2}63216$.
	This will generate a $[ 494, 470 ]_9$ code, which can be verified as a dual-containing code concerning Hermitian inner product.
	Moreover, we have calculated that lower bound on the minimum distance of this code is 5, so there exists a binary QECC with parameters  $[[ 494, 446, \ge5 ]]_{3}$,  which is better than the known $[[494,440,5]]_3$ QECC appeared in \cite{Edel2022}.
	\end{example}
	\section{Conclusion}\label{sec4}
	
	In this paper, we study a class of 2-generator quasi-cyclic codes and provide some sufficient conditions for them to be dual-containing under Hermitian inner product. Furthermore, a lower bound on the minimum distance of these codes is also provided. As an application, many good QECCs over small fields are constructed. In particular, many QECCs are new or have better parameters than those in the literature, so our method of constructing QECCs is valid and feasible. 
	
	However, determining the specific parameters of quasi-cyclic codes with large dimensions and lengths is currently a challenging problem, and even with the use of supercomputers, these problems remain cumbersome.
	Therefore, it will be interesting to investigate more accurate quasi-cyclic code structures and minimize the amount of arithmetic power consumed in the calculation of parameters. 
	We hope that this will attract the interest of scholars in research related to quasi-cyclic codes and together advance this area.

	\bmhead{Acknowledgments}
	
	This work is supported by the National Natural Science Foundation of China
	under Grant No.U21A20428, 11901579,11801564, Natural Science Foundation of
	Shaanxi under Grant No.2021JM-216, 2021JQ-335, 2022JQ-046 and the Graduate
	Scientific Research Foundation of Fundamentals Department of Air
	Force Engineering University.
	
	\section*{Declarations}
	
	The data used or analyzed in this study are available to all and can be requested from the corresponding author if the reader is interested.

	
	\bibliography{sn-bibliography}
	
	
	\end{document}